\documentclass[9pt]{article}
\usepackage{enumitem,amsmath,amssymb,amsthm,fancyhdr,framed, physics,xcolor, graphicx}
\usepackage[margin=1in,headsep=0.2in]{geometry}
\usepackage{hyperref,url}
\newcommand{\OO}{\mathcal{O}_{\mathbf{s}}}
\newcommand{\smallv}{\mathbf{v}}
\newcommand{\V}{\mathbf{V}}

\newcommand{\x}{\mathbf{x}}
\newcommand{\y}{\mathbf{y}}
\newcommand{\z}{\mathbf{z}}

\newcommand{\Gx}{Y^{\good}}
\newcommand{\GV}{\mathcal{V}^{k,\good}}
\newcommand{\Gz}{Z^{-1}(\z)^{\good}}
\newcommand{\good}{\text{good}}

\newtheorem{theorem}{Theorem}
\newtheorem{lemma}{Lemma}

\begin{document}
\title{Optimal Quantum Algorithm for Vector Interpolation}
\author{\textsc{Sophie Decoppet}\thanks{\textsc{Department of Physics and Department of Computer Science, Stanford University}}}
\date{}

\maketitle

\begin{abstract}
    In this paper we study the functions that can be learned through the polynomial interpolation quantum algorithm designed by Childs et al. This algorithm was initially intended to find the coefficients of a multivariate polynomial function defined on finite fields $\mathbb{F}_q$. We extend its scope to vector inner product functions of the form $\OO(\smallv) = \mathbf{s}\cdot\smallv$ where the goal is to find the vector $\mathbf{s} \in \mathbb{F}_q^n$. We examine the necessary conditions on the domain $\mathcal{V}$ of $\OO$ and prove that the algorithm is optimal for such functions. Furthermore, we show that the success probability approaches 1 for large $q$ and large domain order $|\mathcal{V}|.$ Finally, we provide a conservative formula for the number of queries required to achieve this success probability.
\end{abstract}

\section{Introduction}

The task of learning functions is a widespread area of research in theoretical computer science. Algorithms that interpolate a function have numerous applications in theoretical physics, cryptography, scientific computation, and other related fields. Many functions have efficient classical interpolation algorithms, but quantum computers offer the possibility of a significant speed-up. As a result, a burgeoning area of research in quantum computer science is quantum learning algorithms.

The polynomial interpolation problem was initially solved by Childs et al. \cite{univariate} and subsequently generalized by Chen et al. \cite{multivariate}. The objective of the polynomial interpolation problem is to find the coefficients of a polynomial that we have access to through queries to a black box. The quantum algorithm outputs the desired result with probability close to 1 for large domain size and uses significantly less queries than the classical interpolation algorithm. This problem is of particular relevance in cryptography, because it breaks the Shamir secret sharing scheme. The security of this scheme relies on the inability of a classical adversary to find $f(x_{d+1})$ given $f(x_1), \ldots, f(x_d)$ where $f$ is a univariate polynomial of degree $d$. A quantum adversary would be able to interpolate the entire polynomial with the first $d$ values of $f$ using Childs et al.'s algorithm.

In this paper, we study the polynomial interpolation algorithm presented by Childs et al. \cite{univariate} and examine its scope. We find that the polynomial interpolation algorithm can be generalized to interpolate functions defined on finite fields $\mathcal{V} \subseteq \mathbb{F}_q^n$ of the form 
$$\OO(\smallv) = \mathbf{s} \cdot \smallv$$
The goal is to find $\mathbf{s}$, given access to a black box that will output $\OO(\smallv)$ for any $\smallv \in \mathcal{V}$. While the polynomial interpolation algorithm was initially designed to work for a specific input domain $\mathcal{V}$, we show that it works for any domain $\mathcal{V}$ of $\OO$ provided $\mathcal{V}$ follows certain constraints: the field $\mathcal{V}$ must be of size at least $q$ and constructed such that any distinct $n$ vectors in $\mathcal{V}$ are linearly independent. This means that the polynomial interpolation algorithm can be used to interpolate any function of the form $\OO$ even if the type of input $\smallv$ that can be passed into the black box is restricted. We further show that the success probability of the algorithm is $\frac{|R_k|}{q^n}$ where $R_k$ is the range of a function defined in section \ref{section:Zfunc}, and $k$ is the number of queries to the black box. We prove that the algorithm is optimal for these functions, so no other $k$-query algorithm can succeed in finding $\mathbf{s}$ with greater probability. Finally, we bound $|R_k|$ and calculate the query complexity, offering conservative formulas for the number of queries $k$ needed to find $\mathbf{s}$. With $k = \lceil\frac{n\log q}{\log |\mathcal{V}|q}\rceil$ queries, the success probability of the algorithm is $\frac{1}{k!}(1 - O(\frac{1}{\text{min}(q,|\mathcal{V}|)}))$. To achieve a success probability of $1 -O( \frac{1}{\text{min}(q,|\mathcal{V}|)})$, we use $k = \lceil\frac{\log(|\mathcal{V}|q^n)}{2\log(|\mathcal{V}|/|\mathcal{V}_0|)}\rceil$ when $q > |\mathcal{V}|$ and $k = \lceil\frac{(n+1)\log q}{2\log(|\mathcal{V}|/|\mathcal{V}_0|)}\rceil$ when $q < |\mathcal{V}|$, where $\mathcal{V}_0$ is a subset of $\mathcal{V}$ containing every element of $\mathcal{V}$ that has at least one zero entry.

In section \ref{section:paperprelims}, we review the fundamentals of quantum computing, covering qubits, quantum algorithms, and quantum oracles. We also provide an overview of the current research on quantum algorithms that use the same methods as our interpolation algorithm. In section \ref{section:algorithm}, we formalize the problem of vector interpolation and describe the interpolation algorithm. Next, we compute in section \ref{section:performance} the number of queries necessary to succeed with bounded error as well as with probability close to 1 for large $q$ and large function domain. We use two different methods to find a formula for $k$ in each case, following the same techniques as Childs et al. \cite{univariate}. Finally, we show in section \ref{section:optimality} that no other $k$-query algorithm can succeed with probability greater than our interpolation algorithm, using the same approach as Childs et al. \cite{univariate} and Chen et al. \cite{multivariate}. 

\section{Preliminaries}\label{section:paperprelims}
\subsection{Quantum computing overview}

Information on classical computers is stored at the most fundamental level using bits. Each bit can assume either the value 0 or the value 1. Quantum computers function with the quantum analogue of a classical bit: the qubit. A qubit can assume the value 0 or the value 1, just like a classical bit, but can also be in any superposition of those two states. A qubit's state looks like $\ket{\psi} = \alpha \ket{0} + \beta \ket{1}$ for $\alpha, \beta \in \mathbb{C}$. Instead of working in the $\{\ket{0}, \ket{1}\}$ basis, most quantum algorithms work in a different, more practical, basis such as $\{\ket{x} : x \in \mathbb{F}_q\}$. The ability to store a superposition of states in a single qubit is what differentiates quantum computing from classical computing.

A quantum algorithm is a series of unitary transformations that act on the input state. Instead of working with one qubit at a time, a quantum algorithm often acts on multiple registers, each containing separate qubits. For example, an initial state for a quantum algorithm could look like the following two-register uniform superposition over $\mathbb{F}_q$: 
$$\frac{1}{q}\sum_{x, y \in \mathbb{F}_q}\ket{x, y}$$

Quantum algorithms that estimate or learn functions, such as the one studied in this paper, often supplement unitary transformations with queries to an oracle. An oracle is a black box that supplies limited information about a function. Most commonly, an oracle will take an input value $x$ and output $f(x)$. A quantum algorithm can make either adaptive or non-adaptive queries to the oracle. This means that it can make all the queries to the oracle simultaneously (non-adaptive), or it can make the queries sequentially and adjust queries based on the oracle's response to previous queries (adaptive).

Measurement is inherently probabilistic in quantum mechanics. One of the main objectives when designing a quantum algorithm, therefore, is to optimize the success probability. Efficiency can be measured many different ways, for example, number of gates used, size of registers used, or number of queries made to the oracle. In this paper, we will determine the number of queries to the oracle needed to optimize the success probability of our algorithm.

\subsection{Literature review}
The polynomial interpolation algorithm studied in this paper relies on the method of coset states used in the coset identification problem. In this section we provide an overview of algorithms and research conducted on the coset state method used in polynomial interpolation. A more complete taxonomy of the field of quantum algorithms can be found in Gill et al.'s and Stigsson's works \cite{qataxonomy, cryptotaxonomy}.

The initial algorithm presented by Childs et al. \cite{univariate} finds the coefficients of a degree $d$ polynomial function in one variable defined over a finite field $\mathbb{F}_q$. While a classical algorithm would need $d+1$ queries, Childs et al.'s quantum algorithm determines the coefficients of the polynomial with bounded error using $k = \frac{d}{2} + \frac{1}{2}$ queries. Using $k = \frac{d}{2} + 1$ queries, it finds the coefficients with probability approaching 1 for large $q$. This algorithm was later generalized by Chen et al. \cite{multivariate} to interpolate degree $d$ polynomials in $m$ variables defined over finite fields $\mathbb{F}_q$, the reals $\mathbb{R}$, and the complex numbers $\mathbb{C}$. Again, the algorithm provides a significant improvement on the number of queries required compared to the classical case. Classically, we would need ${m+d \choose d}$ queries, whereas the quantum algorithm requires only $\frac{d}{m+d} {m+d \choose d}$ queries, a much more significant jump than in the univariate case.

The structure of the polynomial interpolation algorithm \cite{univariate, multivariate} is based on the Pretty Good Measurement (PGM, \cite{pgm}) approach used in the Hidden Subgroup Problem (HSP, \cite{hsp}). This approach starts out with a uniform superposition over all possible inputs to the black box (i.e. the entire domain of the function evaluated by the oracle). It then performs non-adaptive queries to the black box in order to create a superposition of coset states. A series of entangled measurements then reveals the hidden subgroup. The difference that the polynomial interpolation algorithm introduces is to start out with a uniform superposition over a carefully chosen subset of the function domain. The final state becomes a superposition of coset states, but over only a subset of the possible outputs of the oracle (i.e. a subset of the range of the function evaluated by the oracle). The hope is that a measurement of these entangled states will give the desired result with higher probability than if we had measured a superposition over the entire function range. 

The coset approach used in the polynomial interpolation algorithm is a special case of a more general method studied by Zhandry \cite{zhandry} and Copeland and Pommersheim \cite{symmetricoracle}. Both examine the coset identification problem \cite{symmetricoracle}, called the quantum oracle classification in \cite{zhandry}, and examine its quantum query complexity. The goal of the coset identification problem is to learn the coset structure of an abelian \cite{zhandry} or arbitrary finite \cite{symmetricoracle} group $G$ by querying a black box. More specifically, we are given a group $G$, a subgroup $H \leq G$, and a black box that evaluates $\pi(g)$ for $\pi$ some unitary transformation and $g \in G$. The objective is to find what coset of $H$ $g$ is an element of. Couching their research in terms of character and representation theory, Copeland and Pommersheim show that the optimal quantum algorithm for the coset identification problem is a parallel non-adaptive algorithm. They further prove that any adaptive algorithm is equivalent to a non-adaptive one. Building on Zhandry's work, they provide a general formula for the optimal success probability of a $k$-query quantum algorithm for the coset identification problem. Copeland and Pommersheim's analysis of the coset identification problem reduces to the coset method in the case of polynomial interpolation, where $G$ is the group of polynomial functions \cite{univariate, multivariate}.

The coset identification paradigm can be used to estimate the quantum query complexity of many different learning problems. It was recently applied by Childs et al. \cite{matrixchilds} to the problem of matrix learning. They show that finding the trace of an $n\times n$ matrix with entries in $\mathbb{F}_q$, with probability greater than $\frac{1}{q}$, requires at least $\frac{n}{2}$ queries. Both Zhandry and Copeland and Pommersheim also show how the coset identification paradigm can be applied to the group summation problem (where we want to evaluate the sum $f(1) + \ldots + f(k)$) and the problem of learning a boolean function \cite{boolean}. 

In the next section, we will present the interpolation algorithm and formally define the types of functions it can be used for.  

\section{The interpolation algorithm}\label{section:algorithm}
\subsection{Presentation of the problem}\label{section:presentation}

Let $\mathbb{F}_q$ be a finite field, with $q = p^r$ a power of a prime $p$. The task of learning a polynomial $f$ defined over a finite field $\mathbb{F}_q$ has been studied by Childs et al. \cite{univariate, multivariate}. The concept is straightforward. We have a polynomial $f \in \mathbb{F}_q[x]$ of degree $d$, and we have access to this polynomial through a black box that evaluates $f(x)$ for any input $x \in \mathbb{F}_q$. The objective is to find the coefficients $c_0, \ldots, c_d$. The polynomial interpolation algorithm presented by Childs et al. determines the coefficients with probability $1 - O(\frac{1}{q})$. 

In this paper, we would like to find for which other functions $f$ the polynomial interpolation algorithm \cite{univariate} can be used. More precisely, we focus on vector functions of the form
\begin{align}
    &\OO : \mathcal{V} \to \mathbb{F}_q \\
    &\OO(\smallv) = \mathbf{s} \cdot \smallv
\end{align}
where $\mathcal{V}\subseteq \mathbb{F}_q^n$ is some field of $n$-dimensional vectors. Given a black box that will evaluate $\OO$ on any input $\smallv \in \mathcal{V}$, we would like to use Childs et al.'s interpolation algorithm to find $\mathbf{s} \in \mathbb{F}_q^n$. The objective of this paper is to find for which domain field $\mathcal{V}$ this is possible. We impose a simple condition as described in the theorem below:

\begin{theorem}\label{theorem:maintheorem}
    Let $\mathcal{V} \subseteq \mathbb{F}_q^n$ be a field of size at least $q$ of $n$-dimensional vectors such that any $n$ arbitrarily chosen vectors in $\mathcal{V}$ are linearly independent. Then there exists a $k$-query quantum algorithm for interpolating $\mathbf{s}$ with probability $\frac{|R_k|}{q^n}$, where $R_k \subseteq \mathcal{V}^k \times \mathbb{F}_q^k$ is the range of the function $Z = \sum_{i = 1}^k y_i\smallv_i$ (see section \ref{section:Zfunc}).
\end{theorem}

Polynomial interpolation is a special case of the problem presented above, where we set 
\begin{equation}
    \mathcal{V} = \{(1, x, x^2, \ldots, x^d) : x \in \mathbb{F}_q\}
\end{equation} and $n= d+ 1$, with $d$ the degree of the polynomial. Sections \ref{section:univariate} and \ref{section:multivariate} compare the results found in this paper with the results of \cite{univariate} and \cite{multivariate}.

\subsection{Definitions and notation}
Before presenting the interpolation algorithm, we will start by defining some functions and sets used. First, we recall the behavior of the exponential function $e : \mathbb{F}_q \to \mathbb{C}$ defined on a finite field:
\begin{equation}
    e(z) = e^{2\pi i \text{Tr}(z)}
\end{equation}
where Tr is the trace function 
\begin{equation}
    \text{Tr}(z) = z + z^p + z^{p^2} + \ldots + z^{p^{r-1}}
\end{equation}
Next, we define the quantum Fourier transform over $\mathbb{F}_q$. For any $x \in \mathbb{F}_q$, the quantum Fourier transform of $\ket{x}$ is 
\begin{equation}
    \ket{x} \mapsto \frac{1}{\sqrt{q}}\sum_{y \in \mathbb{F}_q} e(xy)\ket{y}
\end{equation}
Generalizing to any $\x \in \mathbb{F}_q^k$, we define the quantum Fourier transform for $k$-dimensional vectors as follows: 
\begin{equation}
    \ket{\x} \mapsto \frac{1}{\sqrt{q^k}}\sum_{\mathbf{y} \in \mathbb{F}_q^k} e(\x\cdot \mathbf{y})\ket{\mathbf{y}}
\end{equation}

\subsubsection{The $Z$ function}\label{section:Zfunc}
The function $Z$ defined below is used in the interpolation algorithm to store ``trial" values of $\mathbf{s}\cdot \smallv$ for different values of $\mathbf{s}$ and $\smallv$. It is used to create a non-uniform superposition of possible inputs $\smallv \in \mathcal{V}$, as detailed below.

Before defining $Z$, we set forth the following notation. We index a vector $\smallv_i \in \mathcal{V}$ as follows:
\begin{equation}
    \smallv_i = (v_{i,1},v_{i,2},\ldots,v_{i,n})
\end{equation}
This allows us to define the $nk$-dimensional vector $\V$ as $k$ $n$-dimensional vectors:
\begin{align}
    \V &= (\smallv_1, \smallv_2, \ldots, \smallv_k)= (v_{1,1}, v_{1,2}, \ldots, v_{2,1}, v_{2,2}, \ldots, v_{k,1},  v_{k,2}, \ldots, v_{k, n})
\end{align}
Finally, we define the function $Z$:
\begin{align}\label{eq:Z_function}
    &Z : \mathcal{V}^k \times \mathbb{F}_q^k \to \mathbb{F}_q^n\\
    &Z(\V, \y) = \sum_{i = 1}^k y_i\smallv_i
\end{align}

The interpolation algorithm sets up a uniform superposition of vectors, querying the black box $k$ times. A traditional approach would create this superposition over all possible $\smallv$, but our interpolation algorithm starts with a uniform superposition over the pre-image of $Z$ instead. More precisely, it starts with a uniform superposition over the pre-image set $T_k$ and finishes with a uniform superposition over the image set $R_k$, where $T_k$ and $R_k$ are the sets defined below:

\begin{equation}
    R_k = \{Z(\V, \y) : (\V, \y) \in \mathcal{V}^k \times \mathbb{F}_q^k\} \subseteq \mathbb{F}_q^n
\end{equation}
\begin{equation}
    T_k = \{\text{unique set of pre-images of }\z \in R_k\} \subseteq \mathcal{V}^k \times \mathbb{F}_q^k
\end{equation}

\subsection{Phase query}
The interpolation algorithm accesses the function $\OO$ through queries to a black box. On any input $\ket{\smallv, y}$ for $\smallv \in \mathcal{V}$ and $y \in \mathbb{F}_q$, the black box will perform the following operation: 
\begin{equation}
    \ket{\smallv, y} \mapsto \ket{\smallv, y + \OO(\smallv)}
\end{equation}
The black box query is useful to evaluate $\OO$, but instead of storing the value of $\OO$ in the state itself, we want to store it inside the phase of the original state $\ket{\smallv, y}$ in the following way:
\begin{equation}
    \ket{\smallv, y} \mapsto e(y\OO(\smallv))\ket{\smallv,y}
\end{equation}
This allows us to leverage the quantum Fourier transform and to measure the final state in the Fourier basis.

A phase query (query by which we store $\OO$ in the phase) can easily be done by performing an inverse Fourier transform over $\mathbb{F}_q$ on the second query register, querying the black box, then performing a Fourier transform on the second query register. Taking $\smallv \in \mathcal{V}$ and $y \in \mathbb{F}_q$, a phase query performs the following operations:
\begin{align}
    \ket{\smallv, y} &\mapsto \frac{1}{\sqrt{q}}\sum_{\alpha \in \mathbb{F}_q} e(-y\alpha)\ket{\smallv,\alpha} & \text{Inverse Fourier transform}\\
    &\mapsto  \frac{1}{\sqrt{q}}\sum_{\alpha \in \mathbb{F}_q} e(-y\alpha)\ket{\smallv,\alpha + \OO(\smallv)} & \text{Query black box}\\
    &\mapsto  \frac{1}{q}\sum_{\alpha, \beta \in \mathbb{F}_q} e(-y\alpha + \beta(\alpha + \OO(\smallv)))\ket{\smallv,\beta} & \text{Fourier transform}\\
    &= \frac{1}{q}\sum_{\alpha \in \mathbb{F}_q} q\delta_{\alpha,-\OO(\smallv)}\;\;e(-y\alpha)\ket{\smallv,\beta} & \label{eq:fourier1}\\
    &= e(y\OO(\smallv))\ket{\smallv,y}\label{eq:fourier2}
\end{align}
where we used in lines (\ref{eq:fourier1}) and (\ref{eq:fourier2}) the fact that $\sum_{z \in \mathbb{F}_q}e(z(x-y)) = q\delta_{x,y}$. 

Notice how only one standard black box query is performed during a phase query, so computing the query complexity is the same whether we count phase queries or standard queries. 

The interpolation algorithm doesn't perform a single phase query, but rather performs $k$ phase queries in parallel, each on a different register. Generalizing the action of a single phase query, $k$ phase queries in parallel perform the following operation:
\begin{align}
    \ket{\V, \y} &\mapsto e\left(\sum_{i = 1}^k y_i\OO(\smallv_i)\right) \ket{\V,\y} = e\left(\sum_{i = 1}^k y_i(\mathbf{s}\cdot \smallv_i)\right) \ket{\V,\y} \\
    & = e\left(\mathbf{s} \cdot Z(\V, \y) \right) \ket{\V,\y}
\end{align}

\subsection{Algorithm}\label{subsection:algorithm}
The interpolation algorithm has three main steps. We start with a uniform superposition over $T_k$, the set of pre-images of the $Z$ function:
\begin{equation}\label{eq:initialsuperposition}
    \frac{1}{\sqrt{|T_k|}} \sum_{(\V, \y) \in T_k} \ket{\V, \y}
\end{equation}
Next, we perform $k$ phase queries in parallel, which gives us the state 
\begin{equation}
    \frac{1}{\sqrt{|T_k|}} \sum_{(\V, \y) \in T_k} e(\mathbf{s} \cdot Z(\V, \y))\ket{\V, \y}
\end{equation}
Finally, we compute $Z$ in place, performing the uniform transformation $\ket{\V, \y} \to \ket{Z(\V, \y)}$. The final state $\ket{\sigma_{R_k}}$ below is an approximate state which resembles the Fourier transform $\ket{\sigma}$ of $\ket{\mathbf{s}}$
\begin{align}
    \ket{\sigma_{R_k}} &= \frac{1}{\sqrt{|R_k|}} \sum_{\z \in R_k} e(\mathbf{s} \cdot \z)\ket{\z} \label{eq:finalstate}\\
    \ket{\sigma} &= \frac{1}{\sqrt{q^n}} \sum_{\z\in \mathbb{F}_q^n} e(\mathbf{s} \cdot \z) \ket{\z}
\end{align}
Instead of outputing $\ket{\sigma}$, the interpolation algorithm gives us $\ket{\sigma_{R_k}}$. Measuring $\ket{\sigma_{R_k}}$ in the Fourier basis will give us $\mathbf{s}$ with probability
\begin{align}
    |\braket{\sigma_{R_k}}{\sigma}|^2 &= \frac{1}{|R_k|q^n}\left(\sum_{\mathbf{z'} \in R_k}\sum_{\z \in \mathbb{F}_q^n} e(\mathbf{s} \cdot (\z - \z'))\braket{\z'}{\z}\right)^2 \\
    &= \frac{1}{|R_k|q^n} \left(\sum_{\mathbf{z'} \in R_k}\sum_{\z \in \mathbb{F}_q^n} \delta_{\z, \z'}\right)^2 = \frac{|R_k|}{q^n}
\end{align}
In the next section, we will find a lower bound on $|R_k|$ in order to find a lower bound on the success probability of this algorithm.

\section{Performance and query complexity}\label{section:performance}
The lower bound on $|R_k|$ depends on the number of queries $k$ performed. In this section we prove the following theorem: 
\begin{theorem}
    The success probability of the interpolation algorithm described in section \ref{subsection:algorithm} is 
    \begin{enumerate}[label=(\roman*)]
        \item $\frac{1}{k!}\left(1-O\left(\frac{1}{\min(q, |\mathcal{V}|)}\right)\right)$ using $k = \left\lceil\frac{n\log q}{\log |\mathcal{V}|q}\right\rceil$ queries
        
        \item $1 -O\left( \frac{1}{\min(q, |\mathcal{V}|)}\right)$ using 
        \begin{enumerate}
            \item $k = \left\lceil\frac{\log(|\mathcal{V}|q^n)}{2\log(|\mathcal{V}|/|\mathcal{V}_0|)}\right\rceil$ queries when $q > |\mathcal{V}|$ 
            \item $k = \left\lceil\frac{(n+1)\log q}{2\log(|\mathcal{V}|/|\mathcal{V}_0|)}\right\rceil$ queries when $q < |\mathcal{V}|$
        \end{enumerate} where $\mathcal{V}_0 = \{\smallv \in \mathcal{V} : \exists i \text{ s.t. } v_i = 0\}$
    \end{enumerate}
\end{theorem}
Notice that the inverse $k!$ dependence in the first case means that as the size $n$ of the vector $\mathbf{s}$ we want to interpolate increases, the success probability of the algorithm decreases rapidly. Increasing the number of queries $k$ (second case), however, greatly improves the success probability and eliminates the $k$ dependence. 

\subsection{Query complexity for success probability of $\frac{1}{k!}\left(1-O\left(\frac{1}{\text{min}(q, |\mathcal{V}|)}\right)\right)$}\label{section:lowperformance}

In this section, we will prove the following lemma. 
\begin{lemma}\label{lemma:lowsuccessprobability}
    Using $k = \left\lceil\frac{n\log q}{\log |\mathcal{V}|q}\right\rceil$ queries, the success probability of the interpolation algorithm is $\frac{1}{k!}\left(1-O\left(\frac{1}{\min(q, |\mathcal{V}|)}\right)\right)$.
\end{lemma}
\begin{proof}
Our first approach to finding a lower bound on $|R_k|$ is to consider a more restricted range for $Z$, which will map to a subset of $R_k$ and allow us to effectively lower bound $|R_k|$. Instead of considering all of $\mathcal{V}^k \times \mathbb{F}_q^k$, we consider instead $\mathcal{V}^{k,\good} \times Y^{\good}$, where
\begin{equation}
    \GV = \{\V \in \mathcal{V}^k : \smallv_i \neq \smallv_j, \; \forall i \neq j, \}
\end{equation}
\begin{equation}
     \Gx = (\mathbb{F}_q^{\times})^k
\end{equation}
Let $\Gz$ to be the set of pre-image of $\z$ that are contained in $\GV \times \Gx$ such that
\begin{equation}
    \Gz = Z^{-1}(\z) \cap (\GV \times \Gx)
\end{equation}
where 
\begin{equation}
    Z^{-1}(\z) = \{(\V, \y) \in \mathcal{V}^k \times \mathbb{F}_q^k : Z(\V, \y) = \z\}
\end{equation}
is the set of pre-images of any $\z \in \mathbb{F}_q^n$. 

The set $R_k$ is the set of images of $Z$. $|R_k|$ corresponds to the number of distinct images of $Z$, or in other words, the number values of $\z$ such that $Z^{-1}(\z)$ is non-empty. We claim that either $\z$ has no good pre-images, so $|Z^{-1}(\z)^{\good}| = 0$, or that its good pre-images are unique up to permutation of the ``inner" vectors $\smallv_i$, so $|Z^{-1}(\z)^{\good}| = k!$.

To see why this is true, assume that there exist two pairs $(\V, \y)$ and $(\mathbf{W}, \mathbf{y}')$ in $\GV \times \Gx$ such that $Z(\V, \y) = Z(\mathbf{W}, \y')$. Permute the indices of the $k$ inner vectors $\smallv_i$ such that $\smallv_i = \mathbf{w}_i$ for $i \in \{1, \ldots, \ell\}$ and $\smallv_i \neq \mathbf{w}_i$ for $i \in \{\ell + 1, \ldots, k\}$. Now, since $Z(\V, \y) = Z(\mathbf{W}, \y')$, we have that
\begin{equation}
    \sum_{i=1}^k y_i\smallv_i - \sum_{i=1}^k y_i'\mathbf{w}_i = 0
\end{equation}
\begin{equation}\label{eq:independence}
    \sum_{i = 1}^{\ell} (y_i - y_i')\smallv_i + \sum_{i = \ell+1}^k y_i\smallv_i - \sum_{i = \ell+1}^k y_i'\mathbf{w}_i = 0
\end{equation}
There are at most $2k = 2\left\lceil\frac{n\log q}{\log |\mathcal{V}|q}\right\rceil \leq 2\left\lceil\frac{n\log q}{\log q^2}\right\rceil\leq n$ terms in this sum. By assumption, any $n$ distinct vectors in $\mathcal{V}$ are linearly independent. Since the $\smallv_i$ and $\mathbf{w}_i$ in (\ref{eq:independence}) are all distinct, they are linearly independent. So, $y_i = y_i'$ for $i \in \{1, \ldots, \ell\}$ and $y_i = y_i' = 0$ for $i \in \{\ell+1, \ldots, k\}$. We know, however, that $y_i, y_i' \neq 0$ since they are both in $\mathbb{F}_q^{\times}$, so $\ell = k$. If there exist two $(\V, \y), (\mathbf{W}, \y') \in \GV \times \Gx$ such that $Z(\V, \y) = Z(\mathbf{W}, \y')$, then $\V = \mathbf{W}$ up to permutation of the inner vectors $\smallv_i$. In other words, if $\z$ has a good pre-image $(\V,\y) \in \GV \times \Gx$, then it is possible to find all the other good pre-images in $\GV \times \Gx$ by permuting the inner vectors $\smallv_i$ of $\V$. It follows, therefore, that either $|Z^{-1}(\z)^{\good}| = 0$ or $|Z^{-1}(\z)^{\good}| = k!$.

As we saw above, $|R_k|$ is the number of values of $\z$ for which $Z^{-1}(\z)$ is non-empty. Since any non-empty set $Z^{-1}(\z)^{\good}$ has size $k!$, this means that $|R_k|$ is at least equal to
\begin{equation}\label{eq:rkbound}
    |R_k| \geq \frac{1}{k!}\sum_{\z \in \mathbb{F}_q^n} |\Gz|
\end{equation}
To estimate the sum in equation (\ref{eq:rkbound}), notice that 
\begin{equation}\label{eq:Gzsum}
    \sum_{\z \in \mathbb{F}_q^n} \left|\Gz\right| = \left|\GV \right|\left|\Gx\right| 
\end{equation}
We need to find estimates for $\left|\GV \right|$ and $\left|\Gx \right|$. This is straightforward:
\begin{equation}
    \left|\GV \right| = k!{|\mathcal{V}| \choose k} = \frac{|\mathcal{V}|!}{(|\mathcal{V}| -k )!} = |\mathcal{V}|^k\left(1 - O\left(\frac{1}{|\mathcal{V}|}\right)\right)
\end{equation} 
\begin{equation}
    \left|\Gx \right| = (q-1)^k = q^k\left(1 - O\left(\frac{1}{q}\right)\right)
\end{equation}

Plugging these estimates into equation (\ref{eq:Gzsum}), we get:
\begin{align}\label{eq:goodvaluesmean}
    \sum_{\z \in \mathbb{F}_q^n} \left|\Gz\right| &= \left|\GV \right|\left|\Gx\right|  \\
    &= (|\mathcal{V}|q)^k\left(1 - O\left(\frac{1}{|\mathcal{V}|}\right)\right)\left(1 - O\left(\frac{1}{q}\right)\right)\\
    &= (|\mathcal{V}|q)^k\left(1 - O\left(\frac{1}{\text{min}(q,|\mathcal{V}|)}\right)\right) \label{eq:zsum}
\end{align}
Now, recall that we set $k = \left\lceil\frac{n\log q}{\log |\mathcal{V}|q}\right\rceil$, so $(|\mathcal{V}|q)^k = q^n$. Plugging this and equation (\ref{eq:zsum}) into equation (\ref{eq:rkbound}), we get: 
\begin{align}
    |R_k| &\geq\frac{q^n}{k!}\left(1 - O\left(\frac{1}{\text{min}(q,|\mathcal{V}|)}\right)\right)
\end{align}
It follows that the probability of success of the algorithm is, as desired,
\begin{align}
    \frac{|R_k|}{q^n} &\geq \frac{1}{k!}\left(1 - O\left(\frac{1}{\text{min}(q,|\mathcal{V}|)}\right)\right)
\end{align}
\end{proof}

\subsection{Query complexity for success probability of $1-O\left(\frac{1}{\text{min}(q, |\mathcal{V}|)}\right)$}\label{section:highperformance}

In this section, we will prove the following lemma. 
\begin{lemma}\label{lemma:highperformance}
    The success probability of the interpolation algorithm is $\left(1-O\left(\frac{1}{\min(q, |\mathcal{V}|)}\right)\right)$ when we use
    \begin{enumerate}[label=(\roman*)]
        \item $k = \left\lceil\frac{\log(|\mathcal{V}|q^n)}{2\log(|\mathcal{V}|/|\mathcal{V}_0|)}\right\rceil$ when $q > |\mathcal{V}|$ \\
        \item $k = \left\lceil\frac{(n+1)\log q}{2\log(|\mathcal{V}|/|\mathcal{V}_0|)}\right\rceil$ when $q < |\mathcal{V}|$
    \end{enumerate}
    where $|\mathcal{V}_0|$ is the number of vectors in $\mathcal{V}$ that have at least one zero entry.
\end{lemma}
\begin{proof}
Recall from Section \ref{section:lowperformance} that $|R_k|$ is the number values of $\z \in \mathbb{F}_q^n$ for which $|Z^{-1}(\z)| \neq 0$. Therefore, under uniform distribution of $\z \in \mathbb{F}_q^n$, we can estimate the success probability as
\begin{equation}\label{eq:secondmoment1}
    \frac{|R_k|}{q^n} = 1 - \text{Pr}[|Z^{-1}(\z)| = 0]
\end{equation}
Our goal is to lower bound $|R_k|$, which requires finding an upper bound for $\text{Pr}[|Z^{-1}(\z)| = 0]$. To do so, we will use the Chebyshev inequality, so that 
\begin{equation}
    \text{Pr}[|Z^{-1}(\z)| = 0] \leq \frac{\sigma^2}{\mu^2}
\end{equation}
The mean of $|Z^{-1}(\z)|$ is 
\begin{align}
    \mu &= \frac{1}{q^{n}} \sum_{\z \in \mathbb{F}_q^n} |Z^{-1}(\z)| = \frac{1}{q^{n}} |\mathcal{V}^k||\mathbb{F}_q^k| = \frac{(|\mathcal{V}|q)^k}{q^{n}}
\end{align}
To calculate the variance we first need to calculate the second moment. \\Recall that $\delta_{x, y} = \begin{cases} 1 & \text{if } x=y \\ 0 & \text{otherwise}\end{cases}$
\begingroup
\allowdisplaybreaks
\begin{align}
    \sum_{\z\in \mathbb{F}_q^n} |Z^{-1}(\z)|^2 &= 
    \sum_{\substack{\V, \mathbf{W} \in \mathcal{V}^k \\ \y, \y' \in \mathbb{F}_q^k}} \delta_{Z(\V, \y), Z(\mathbf{W}, \y')} \\
    &= \sum_{\substack{\V, \mathbf{W} \in \mathcal{V}^k \\ \y, \y' \in \mathbb{F}_q^k}}\frac{1}{q^n}\sum_{\mathbf{t}\in \mathbb{F}_q^n} e(\mathbf{t}\cdot (Z(\V, \y) - Z(\mathbf{W}, \y'))) \\
    &= \frac{1}{q^n}\sum_{\substack{\V, \mathbf{W} \in \mathcal{V}^k \\ \y, \y' \in \mathbb{F}_q^k}} e(\mathbf{0}\cdot (Z(\V, \y) - Z(\mathbf{W}, \y')))  \notag \\
    &\quad\quad\quad + \sum_{\substack{\V, \mathbf{W} \in \mathcal{V}^k \\ \y, \y' \in \mathbb{F}_q^k}}\frac{1}{q^n}\sum_{\mathbf{t}\in \mathbb{F}_q^n \backslash \mathbf{0}} e(\mathbf{t}\cdot (Z(\V, \y) - Z(\mathbf{W}, \y')))\\
    &= \frac{(|\mathcal{V}|q)^{2k}}{q^n} + \frac{1}{q^n}\sum_{\mathbf{t}\in \mathbb{F}_q^n \backslash \mathbf{0}} \sum_{\substack{\V, \mathbf{W} \in \mathcal{V}^k \\ \y, \y' \in \mathbb{F}_q^k}} e\left(\mathbf{t}\cdot\sum_{i=1}^k (y_i\mathbf{v_i} - y_i'\mathbf{w_i})\right)\\
    &= \frac{(|\mathcal{V}|q)^{2k}}{q^n} + \frac{1}{q^n}\sum_{\mathbf{t}\in \mathbb{F}_q^n \backslash \mathbf{0}} \sum_{\substack{\V, \mathbf{W} \in \mathcal{V}^k \\ \y, \y' \in \mathbb{F}_q^k}} \prod_{i=1}^ke\left( \mathbf{t} \cdot (y_i\mathbf{v_i} - y_i'\mathbf{w_i})\right)\\
    %&= \frac{(|\mathcal{V}|q)^{2k}}{q^n} + \frac{1}{q^n}\sum_{\mathbf{t}\in \mathbb{F}_q^n \backslash \mathbf{0}} \left(\sum_{\substack{\V \in \mathcal{V}^k \\ \x \in \mathbb{F}_q^k}} \prod_{i=1}^ke\left(  x_i(\mathbf{t} \cdot \mathbf{v_i})\right)\right)^2\\
    &= \frac{(|\mathcal{V}|q)^{2k}}{q^n} + \frac{1}{q^n}\sum_{\mathbf{t}\in \mathbb{F}_q^n \backslash \mathbf{0}} \left(\sum_{\substack{\smallv \in \mathcal{V} \\ y\in \mathbb{F}_q}} e\left(y(\mathbf{t} \cdot \smallv)\right)\right)^{2k} \\
    &=\frac{(|\mathcal{V}|q)^{2k}}{q^n} + \frac{1}{q^n}\sum_{\mathbf{t}\in \mathbb{F}_q^n \backslash \mathbf{0}} \left(q\sum_{\smallv \in \mathcal{V}} \delta_{\mathbf{t}\cdot\smallv,0}\right)^{2k}\\
    &= \frac{(|\mathcal{V}|q)^{2k}}{q^n} + \frac{q^{2k}}{q^n}\sum_{\mathbf{t}\in \mathbb{F}_q^n \backslash \mathbf{0}} \left(\sum_{\smallv \in \mathcal{V}} \delta_{\mathbf{t}\cdot\smallv,0}\right)^{2k}\label{eq:V0sumestimate}
\end{align}
The next step is to upper bound the second term in equation (\ref{eq:V0sumestimate}). Notice that $\mathbf{t}\cdot \smallv=0$ if and only if for every $i \in \{1, \ldots, n\}$, one or both of $t_i$ and $v_i$ are equal to zero. Let $\mathcal{V}_0 = \{\smallv \in \mathcal{V} : \exists i \text{ s.t. } v_i = 0\}$, the subset of $\mathcal{V}$ containing all the elements of $\mathcal{V}$ that have at least one entry equal to 0. We can upper bound the second term in the following way: 
\begin{align}\label{eq:highperformance2}
    \frac{q^{2k}}{q^n}\sum_{\mathbf{t}\in \mathbb{F}_q^n \backslash \mathbf{0}} \left(\sum_{\smallv \in \mathcal{V}} \delta_{\mathbf{t}\cdot\smallv,0}\right)^{2k} 
    &\leq \frac{q^{2k}}{q^n}\sum_{\mathbf{t}\in \mathbb{F}_q^{n} \backslash \mathbf{0}} |\mathcal{V}_0|^{2k}= \frac{(q^n -1)(|\mathcal{V}_0|q)^{2k}}{q^n} \\
    &\leq (|\mathcal{V}_0|q)^{2k}
\end{align}%
\endgroup
Thus the variance is
\begin{align}
    \sigma^2 &= \frac{1}{q^n} \sum_{\z\in \mathbb{F}_q^n} |Z^{-1}(\z)|^2 - \mu^2 \\
    &\leq \frac{(|\mathcal{V}|q)^{2k}}{q^{2n}} + \frac{(|\mathcal{V}_0|q)^{2k}}{q^n} - \frac{(|\mathcal{V}|q)^{2k}}{q^{2n}}  = \frac{(|\mathcal{V}_0|q)^{2k}}{q^n} 
\end{align}
Applying Chebyshev's inequality, we find: 
\begin{align}
    \text{Pr}[|Z^{-1}(\z)| = 0 ] &\leq \frac{\sigma^2}{\mu^2} \leq \frac{(|\mathcal{V}_0|q)^{2k}}{q^n} \frac{q^{2n}}{(|\mathcal{V}|q)^{2k}} = q^n\left(\frac{|\mathcal{V}_0|}{|\mathcal{V}|}\right)^{2k}
\end{align}
Now, we look at the two cases where $q < |\mathcal{V}|$ and $q > |\mathcal{V}|$.
\begin{enumerate}[label=(\roman*)]
    \item \textbf{Case $q < |\mathcal{V}|$}\\
In this case, we set 
\begin{equation}\label{eq:ksmallq}
    k = \left\lceil\frac{(n+1)\log q}{2\log(|\mathcal{V}|/|\mathcal{V}_0|)}\right\rceil
\end{equation}
for which $\text{Pr}[|Z^{-1}(\z)| = 0] \leq O\left(\frac{1}{q}\right)$.
    \item \textbf{Case $|\mathcal{V}| < q$}\\
In this case, we set 
\begin{equation}\label{eq:ksmallV}
    k = \left\lceil\frac{\log(|\mathcal{V}|q^n)}{2\log(|\mathcal{V}|/|\mathcal{V}_0|)}\right\rceil
\end{equation}
for which $\text{Pr}[|Z^{-1}(\z)| = 0] \leq O\left(\frac{1}{|\mathcal{V}|}\right)$.
\end{enumerate}
In each case, we get 
\begin{equation}
    \text{Pr}[|Z^{-1}(\z)| = 0] \leq O\left(\frac{1}{\text{min}(q, |\mathcal{V}|)}\right)
\end{equation}
Therefore, the success probability of the algorithm is, as desired,
\begin{align}
    \frac{|R_k|}{q^n} = 1 - \text{Pr}[|Z^{-1}(\z)| = 0] \geq 1 - O\left(\frac{1}{\text{min}(q, |\mathcal{V}|)}\right)
\end{align}
\end{proof}
The formulas (\ref{eq:ksmallq}) and (\ref{eq:ksmallV}) are exact when the term inside the sum in line (\ref{eq:V0sumestimate}) is constant for all $\mathbf{t} \in \mathbb{F}_q^n$, as we will see in the example of univariate polynomial interpolation in section \ref{section:univariate}. If, on the other hand, the term inside the sum is not constant for all $\mathbf{t} \in \mathbb{F}_q^n$, then the values of $k$ given in equations (\ref{eq:ksmallq}) and (\ref{eq:ksmallV}) are conservative. The limit value of $k$ such that the success probability is $1 - O(\frac{1}{\text{min}(q,|\mathcal{V}|)})$ could be smaller, and must be found either by directly evaluating the sum in equation (\ref{eq:V0sumestimate})(although this might prove to be difficult), or by another method entirely, such as using algebraic geometry as in the multivariate polynomial case \cite{multivariate}.

\section{Special cases of vector interpolation}
\subsection{Univariate polynomial interpolation}\label{section:univariate}
The interpolation algorithm presented in section \ref{subsection:algorithm} was initially created by Childs et al.\cite{univariate} to find the coefficients of a univariate polynomial $f \in \mathbb{F}_q[x]$. Indeed, polynomial interpolation is a special case of vector interpolation. To see why this is the case, we will find the appropriate vector $\ket{s}$ and field $\mathcal{V}$. 

Let $n-1$ be the degree of the polynomial, so $f(x) = c_0 + c_1x + \ldots + c_{n-1}x^{n-1}$. Set $\\{\ket{s} = \ket{c} = (c_0, \ldots, c_{n-1})}$ to be the coefficient vector and let $\mathcal{V} = \{(1, x, x^2, \ldots, x^{n-1}) \; : \; x \in \mathbb{F}_q\}$. Then we have 
\begin{align}
    \OO(\smallv) &= \mathbf{s} \cdot \smallv = (c_0, c_1, \ldots, c_{n-1}) \cdot (1, x, \ldots, x^{n-1}) \\
    &= c_0 + c_1x + \ldots + c_{n-1}x^{n-1} = f(x)
\end{align}
This field $\mathcal{V}$ fulfills the condition set forth in theorem \ref{theorem:maintheorem}. Indeed, consider the Vandermonde matrix, where each column is an arbitrary element of $\mathcal{V}$ and each of the $n$ columns are distinct:
\begin{equation}\label{eq:vandermondematrix}
    \text{Vand}_n = \begin{pmatrix} 
1 & 1 & \ldots & 1 \\
x_1 & x_2 & \ldots & x_{n} \\
x_1^2 & x_2^2 & \ldots & x_{n}^2 \\
\vdots & \vdots & \ddots & \vdots \\
x_1^{n-1} & x_2^{n-1} & \ldots & x_{n}^{n-1}
\end{pmatrix}
\end{equation}
This matrix is invertible, which means that the $n$ columns are linearly independent, as wanted.

Childs et al. determined in \cite{univariate} that a value of $k = \frac{d}{2} + \frac{1}{2}$ for $d$ odd would achieve a success probability of $\frac{1}{k!}(1 - O(\frac{1}{q}))$ and a value of $k = \frac{d}{2} + 1$ for $d$ even would achieve a success probability of $1 - O(\frac{1}{q})$. We will check these values with the values of $k$ found above in section \ref{section:lowperformance} and section \ref{section:highperformance}.

First, we claim that $|\mathcal{V}| = q$ and $|\mathcal{V}_0| = 1$. Indeed, there is a one-to-one correspondence between $\smallv \in \mathcal{V}$ and $\mathbb{F}_q$ where $(1, x, \ldots, x^{n-1})\longleftrightarrow x$. Therefore, $|\mathcal{V}| = |\mathbb{F}_q| = q$. To see that $|\mathcal{V}_0| = 1$, recall that $\mathcal{V}_0$ is the set of vectors in $\mathcal{V}$ that have at least one entry equal to 0. If any entry $v_i = x^{i-1}$ of $\smallv$ is equal to 0 for $i \in \{1, \ldots, n\}$, then $x = 0$, so $\smallv = (1, 0, \ldots, 0)$. It follows then that $|\mathcal{V}_0| = 1$.

Now, we plug these values into the formulas for $k$ found in section  \ref{section:lowperformance}:
\begin{align}
    k = \left\lceil\frac{n\log q}{\log |\mathcal{V}|q}\right\rceil = \left\lceil\frac{n\log q}{\log q^2}\right\rceil = \left\lceil\frac{n}{2}\right\rceil = \left\lceil\frac{d+1}{2}\right\rceil
\end{align}
For $d$ odd, this is equal to $\frac{d}{2} + \frac{1}{2}$, as wanted.

Finally, we check the formulas for $k$ found in section \ref{section:highperformance}. Since $|\mathcal{V}| = q$, we check both formulas (\ref{eq:ksmallq}) and (\ref{eq:ksmallV}):
\begin{align}
    k_{q < |\mathcal{V}|} &= \left\lceil\frac{(n+1)\log q}{2\log(|\mathcal{V}|/|\mathcal{V}_0|)}\right\rceil= \left\lceil\frac{(n+1)\log q}{2\log q}\right\rceil = \left\lceil\frac{n+1}{2}\right\rceil = \left\lceil\frac{d}{2} + 1\right\rceil
\end{align}
\begin{align}
    k_{|\mathcal{V}| < q} &= \left\lceil\frac{\log(|\mathcal{V}|q^n)}{2\log(|\mathcal{V}|/|\mathcal{V}_0|)}\right\rceil = \left\lceil\frac{\log(q^{n+1})}{2\log q}\right\rceil = \left\lceil\frac{n+1}{2}\right\rceil = \left\lceil\frac{d}{2} + 1\right\rceil
\end{align}
For $d$ even, $k = \frac{d}{2} + 1$, as wanted. Therefore, our formulas for the number of queries $k$ match the previously determined values found by Childs et al. in \cite{univariate}.

\subsection{Multivariate polynomial interpolation}\label{section:multivariate}
The vector interpolation algorithm can also be used to interpolate multivariate polynomials $\\{f \in \mathbb{F}_q[x_1, \ldots, x_m]}$, as studied by Chen et al. in \cite{multivariate}. To see why this is the case, we will find the appropriate vector $\ket{s}$ and set $\mathcal{V}$. 

Let $d$ be the degree of the polynomial $f$. Set 
\begin{equation}
\mathcal{V} = \{(1, x_1, \ldots, x_m, x_1x_2, \ldots, x_1^{d-1}x_2, \ldots) : x_1, x_2, \ldots, x_m \in \mathbb{F}_q\}
\end{equation}
where each element of $\mathcal{V}$ is a vector containing all the monomials formed by $x_1, \ldots, x_m \in \mathbb{F}_q$ of degree at most $d$. For example, setting $m=2$ and $d=3$, the set $\mathcal{V}$ would be 
\begin{equation}
    \mathcal{V} = \{(1, x_1, x_2, x_1x_2, x_1^2, x_2^2, x_1^2x_2, x_2^2x_1, x_1^3, x_2^3) : x_1, x_2 \in \mathbb{F}_q\}
\end{equation}
We claim that $|\mathcal{V}| = q^{m}$. To see why this is the case, notice that each vector $\smallv \in \mathcal{V}$ is uniquely determined by $x_1, x_2, \ldots, x_m$, so $|\mathcal{V}| = |\{(x_1, \ldots, x_m) : x_1, \ldots, x_m \in \mathbb{F}_q\}| = q^m$. 

We further claim that $n = {m+d \choose d}$. This is the number of monomials of degree at most $d$ that can be formed by $m$ variables.

Finally, we compute the value of $|\mathcal{V}_0|$. We will find an explicit formula then bound it above and below. If a vector $\smallv \in \mathcal{V}$ has a 0 entry, then at least one of $x_1, \ldots, x_m$ is 0. Therefore to find $|\mathcal{V}_0|$, we sum over the number $i$ of zeroes in $(x_1, \ldots, x_m)$. For each $i$ we count the number of ways to distribute the zeroes among $(x_1, \ldots, x_m)$ and the number of possible values for the remaining non zero entries in $(x_1, \ldots, x_m)$, so
\begin{align}
    |\mathcal{V}_0| = \sum_{i=1}^m {m \choose i}q^{m - {m \choose i}} 
\end{align}
Instead of evaluating this sum explicitly, we will bound it above and below to find a range for $k$. We will then verify that the $k$ described in \cite{multivariate} falls within the calculated range. 
The smallest value of ${m \choose i}$ is $1$, so an upper bound for $|\mathcal{V}_0|$ is
\begin{align}\label{eq:v0size}
    |\mathcal{V}_0| &\leq q^{m-1}\sum_{i=1}^m {m \choose i} \\
    &= q^{m-1}(2^m - 1)
\end{align}
The largest value of ${m \choose i}$ is upper bounded by $m^m$, so a lower bound for $|\mathcal{V}_0|$ is
\begin{align}
    |\mathcal{V}_0| &\geq q^{m-m^m}\sum_{i=1}^m {m \choose i} \\
    &= q^{m-m^m}(2^m - 1)
\end{align}
In our case $q < |\mathcal{V}|$, so we use the formula for $k_{q < |\mathcal{V}|}$ in equation (\ref{eq:ksmallq}). Calculating the upper bound, we have
\begin{align}
    k &= \left\lceil\frac{(n+1)\log q}{2\log(|\mathcal{V}|/|\mathcal{V}_0|)}\right\rceil  \leq \left\lceil\frac{(n+1)\log q}{2\log(q^m/(q^{m-1}(2^m-1)))}\right\rceil\\
    &= \left\lceil\frac{(n+1)\log q}{2\log(q/(2^m -1))}\right\rceil \leq \left\lceil\frac{(n+1)q^m}{2}\right\rceil
\end{align}
Similarly, calculating the lower bound, we have
\begin{align}
    k &= \left\lceil\frac{(n+1)\log q}{2\log(|\mathcal{V}|/|\mathcal{V}_0|)}\right\rceil  \geq \left\lceil\frac{(n+1)\log q}{2\log(q^m/(q^{m-m^m}(2^m-1)))}\right\rceil\\
    &= \left\lceil\frac{(n+1)\log q}{2\log(q^{m^m}/(2^m -1))}\right\rceil \geq \left\lceil\frac{(n+1)\log q}{2\log q^{m^m}}\right\rceil\\
    &= \left\lceil\frac{n+1}{2m^m}\right\rceil
\end{align}
To summarize, we have the following very conservative bounds for $k$: 
\begin{equation}
    \left\lceil\frac{n+1}{2m^m}\right\rceil \leq k \leq \left\lceil\frac{(n+1)q^m}{2}\right\rceil
\end{equation}
The value of $k$ given in \cite{multivariate} is $\frac{dn}{m+d}$. Since $\frac{1}{2m^m} \leq \frac{d}{m+d}\leq \frac{q^m}{2}$, this value of $k$ is within the bounds, as wanted. To get a more precise value of $k$ from formula (\ref{eq:ksmallq}), we can directly evaluate equation (\ref{eq:v0size}). This will still produce an overestimate of $k$ since $\mathcal{V}_0$ was used to upper bound the sum in line (\ref{eq:V0sumestimate}). Ultimately, a direct evaluation of the sum in line (\ref{eq:V0sumestimate}) will produce a tight bound on the number of queries $k$ needed to interpolate a multivariate polynomial.

%\begin{align}
    %&= \left\lceil\frac{(n+1)\log q}{2\log(q^m/(2^m - 1))}\right\rceil\label{eq:kmultivariate1}\\
    %&= \left\lceil\frac{(n+1)\log q}{2m\log(q/2)}\right\rceil\label{eq:kmultivariate2}
%\end{align}
%Where the step from \ref{eq:kmultivariate1} to \ref{eq:kmultivariate2} works for $q \geq 5$.

Notice that we can reduce a multivariate polynomial to a univariate polynomial by setting $x_i = x_1^{d^{i-1} + d^{i-2} + \ldots + 1}$.  For example, the multivariate polynomial of degree $d = 2$ with $m=3$ variables is
\begin{equation}
    f(x_1, x_2, x_3) = c_0 + c_1x_1 + c_2x_2 + c_3x_3 + c_4x_1x_2 + c_5x_1x_3 + c_6x_2x_3 + c_7x_1^2 + c_9x_2^2 + c_{10}x_3^2
\end{equation}
It becomes, under this transformation, 
\begin{equation}
    f(x_1, x_1^{3}, x_1^{7}) = c_0 + c_1x_1 + c_2x_1^3 + c_3x_1^7 + c_4x_1^4 + c_5x_1^8 + c_6x_1^{10} + c_7x_1^2 + c_9x_1^6 + c_{10}x_1^{14}
\end{equation}
Under this transformation a degree $d$ multivariate polynomial $f$ with $m$ variables becomes a degree $d^{m} + d^{m-1} + \ldots +d$ univariate polynomial. Using the univariate polynomial interpolation algorithm, we can determine the coefficients of $f$ with $k = \frac{d^{m} + d^{m-1} + \ldots +d}{2} + \frac{1}{2}$ queries. This is smaller than the classical number of queries ${m+d\choose d}$ but larger than the multivariate polynomial interpolation algorithm, where the number of queries needed is $\frac{d}{m+d}{m+d \choose d}$.

\section{Optimality}\label{section:optimality}
In this section, we will show that the algorithm presented in section \ref{subsection:algorithm} is optimal for vector interpolation. Specifically, this means that no other $k$-query algorithm will succeed in finding $\mathbf{s}$ with probability greater than $|R_k|/q^n$. We will use the following two lemmas. 

\begin{lemma}\cite{univariate} \label{lemma:childsoptimality}
    Suppose we are given a state $\ket{\phi_c}$ with $c \in C$ chosen uniformly at random from some set $C$. Then the probability of correctly determining $c$ with some orthogonal measurement is at most \normalfont $$\frac{\text{dim span}\{\ket{\phi_c} : c \in C\}}{|C|}$$ 
\end{lemma}

\begin{lemma}\label{lemma:optimality}
    Let $\ket{\phi_c}$ be the state of any quantum algorithm after $k$ queries, where the black box contains $c \in \mathbb{F}_q^n$. Then \normalfont $\text{dim span}\{\ket{\phi_c} : c \in C\} \leq |R_k|$.
\end{lemma} 

Lemma \ref{lemma:childsoptimality} was proven by Childs et al. \cite{univariate}, and we provide a proof of lemma \ref{lemma:optimality} below. 

\begin{proof}[Proof of Lemma \ref{lemma:optimality}.]
Consider the state space $\ket{\smallv,y,\boldsymbol{\xi}}$, with $\ket{\smallv} \in \mathbb{F}_q^n$, and $\ket{y}\in \mathbb{F}_q$, and $\ket{\boldsymbol{\xi}}$ a workspace register of arbitrary size indexed by some index set $I$. Let $A$ be an arbitrary $k$-query algorithm $A = U_{k}Q_cU_{k-1}Q_c\ldots U_1Q_cU_0$ acting on the state space $\ket{\smallv,y,\boldsymbol{\xi}}$, where each $U_i$ is a unitary transformation and $Q_c$ is a phase query acting on the first two registers as follows:
\begin{equation}
    Q_c\ket{\smallv, y, \boldsymbol{\xi}} \mapsto e(y\OO(\smallv))\ket{\smallv,y,\boldsymbol{\xi}}
\end{equation}
We start with the initial state $\ket{\smallv_0, y_0, \boldsymbol{\xi}} = \ket{\mathbf{0}, 0, \mathbf{0}}$, then act with $A$ as follows:
\begingroup
\allowdisplaybreaks
\begin{align}
    \ket{\phi}_c &= A\ket{\smallv_0, y_0, \boldsymbol{\xi}_0} \\
    &= U_{k}Q_cU_{k-1}Q_c\ldots U_1Q_cU_0\ket{\smallv_0, y_0, \boldsymbol{\xi}_0} \\
    &= U_{k}Q_cU_{k-1}Q_c\ldots U_1Q_c\left(\sum_{\substack{\smallv_1 \in \mathbb{F}_q^n \\ y_1\in \mathbb{F}_q \\ \boldsymbol{\xi}_1 \in I}}\ket{\smallv_1, y_1, \boldsymbol{\xi}_1}\bra{\smallv_1, y_1, \boldsymbol{\xi}_1}\right)U_0\ket{\smallv_0, y_0, \boldsymbol{\xi}_0} \\
    &= \sum_{\substack{\smallv_1 \in \mathbb{F}_q^n \\y_1\in \mathbb{F}_q\\ \boldsymbol{\xi}_1 \in I}} e(y_1\OO(\smallv_1))\;\;U_{k}Q_cU_{k-1}Q_c\ldots U_1\ket{\smallv_1, y_1, \boldsymbol{\xi}_1}\bra{\smallv_1, y_1, \boldsymbol{\xi}_1}U_0\ket{\smallv_0, y_0, \boldsymbol{\xi}_0}\\
    &\;\;\vdots \notag \\
    &= \sum_{\substack{\smallv_{k+1} \in \mathbb{F}_q^n \\y_{k+1}\in \mathbb{F}_q \\ \boldsymbol{\xi} \in I^{k+1}}}\sum_{\substack{\V \in (\mathbb{F}_q^{n})^k \\ \y \in \mathbb{F}_q^k}}e\left(\sum_{i = 1}^ky_i\OO(\smallv_i)\right)\left(\prod_{i =0}^k \bra{\smallv_{i+1}, y_{i+1}, \boldsymbol{\xi}_{i+1}}U_i\ket{\smallv_i, y_i, \boldsymbol{\xi}_i}\right)\ket{\smallv_{k+1}, y_{k+1}, \boldsymbol{\xi}_{k+1}} \\
    &= \sum_{\z\in \mathbb{F}_q^n} e(\mathbf{s} \cdot \z) \sum_{(\V, \y) \in Z^{-1}(\z)} \sum_{\substack{\smallv_{k+1} \in \mathbb{F}_q^n \\y_{k+1}\in \mathbb{F}_q \\ \boldsymbol{\xi}\in I^{k+1}}}\left(\prod_{i =0}^k \bra{\smallv_{i+1}, y_{i+1}, \boldsymbol{\xi}_{i+1}}U_i\ket{\smallv_i, y_i, \boldsymbol{\xi}_i}\right)\ket{\smallv_{k+1}, y_{k+1}, \boldsymbol{\xi}_{k+1}} \label{eq:finaloptimality}
\end{align}
\endgroup
Setting 
\begin{equation}
    \ket{\psi_{\z}} = \sum_{(\V, \y) \in Z^{-1}(\z)} \sum_{\substack{\smallv_{k+1} \in \mathbb{F}_q^n \\y_{k+1}\in \mathbb{F}_q \\ \boldsymbol{\xi} \in I^{k+1}}}\left(\prod_{i =0}^k \bra{\smallv_{i+1}, y_{i+1}, \boldsymbol{\xi}_{i+1}}U_i\ket{\smallv_i, y_i, \boldsymbol{\xi}_i}\right)\ket{\smallv_{k+1}, y_{k+1}, \boldsymbol{\xi}_{k+1}} 
\end{equation}
we can rewrite the final state in line (\ref{eq:finaloptimality}) as
\begin{equation}\label{eq:finalstateoptimality}
    \ket{\phi_c} = \sum_{\z\in \mathbb{F}_q^n} e(\mathbf{s} \cdot \z)  \ket{\psi_{\z}}
\end{equation}
The span of $\{\ket{\psi_{\z}} : \z \in R_k\}$ has dimension at most $|R_k|$. From equation (\ref{eq:finalstateoptimality}) we see that the final state of the system after $A$ is an element of span$\{\ket{\psi_{\z}} : \z \in R_k\}$, so we get, as wanted,
\begin{equation}
    \text{dim span}\{\ket{\phi_c} : c \in \mathbb{F}_q^n\} \leq \text{dim span}\{\ket{\psi_{\z}} : \z \in R_k\} \leq |R_k|
\end{equation}
\end{proof}
By combining lemma \ref{lemma:childsoptimality} and lemma \ref{lemma:optimality}, we see that the probability of correctly determining $\mathbf{s}$ chosen uniformly at random with some orthogonal measurement is at most 
\begin{equation}
    \frac{\text{dim span}\{\ket{\phi_{\mathbf{s}}} : \mathbf{s} \in \mathbb{F}_q^n\}}{|\mathbb{F}_q^n| } \leq \frac{|R_k|}{q^n}
\end{equation}
Determining $\mathbf{s}$ from a uniform distribution over $\mathbb{F}_q^n$ is the average case scenario, so any $k$-query algorithm will always determine $\mathbf{s}$ with a probability of at most $|R_k|/q^n$ (the worst case scenario will succeed with lower probability than the average case). Therefore, the quantum interpolation algorithm described in this paper is optimal for vector interpolation. 

\section{Conclusion and open problems}
In this paper we examined the quantum polynomial interpolation algorithm and broadened its scope to interpolate certain types of vector inner product functions. We found a conservative bound on the query complexity. We conclude by discussing avenues for future work. 

The query complexity computed in section \ref{section:highperformance} is an overestimate in most cases. Is it possible to find a tighter bound on the number of queries $k$? Our formula for $k$ required us to evaluate large sums in the multivariate polynomial case. Is it possible to find an overestimate that will simplify calculations?

We evaluated the query complexity of the algorithm, but an opportunity for future research would be to compute the gate complexity of the interpolation algorithm. This was already calculated in the univariate polynomial case by Childs et al. \cite{univariate}, so we could generalize their study to the vector interpolation problem. 

Finally, our interpolation algorithm is a special case of the coset identification problem. A closer study of this larger framework could help identify other functions for which the interpolation algorithm is optimal.

\section*{Acknowledgements}
We thank Dan Boneh and Patrick Hayden for their guidance and helpful discussions.

\end{document}